\documentclass{article}

\usepackage[disable]{todonotes}

\usepackage{hyperref}   
\usepackage{microtype}  
\usepackage{graphicx}   
\usepackage{subcaption} 
\usepackage{booktabs}   
\usepackage{tabularx}   
\usepackage{adjustbox}  
\usepackage{xspace}     
\usepackage{comment}    
\usepackage{enumitem}   
\usepackage{multirow}   
\usepackage{amssymb}    
\usepackage{amsfonts}   
\usepackage{mathtools}  
\usepackage{amsthm}     
\usepackage{dsfont}     
\usepackage[table]{xcolor} 
\usepackage[algo2e,noend,ruled,linesnumbered]{algorithm2e} 
\usepackage{scrextend}  
\usepackage[sort&compress,nameinlink,noabbrev,capitalize]{cleveref} 

\usepackage{tikz} 
\usetikzlibrary{arrows}
\usetikzlibrary{shapes.multipart}
\usetikzlibrary{positioning}
\usetikzlibrary{intersections}
\usetikzlibrary{calc}
\usetikzlibrary{decorations.pathreplacing}
\usetikzlibrary{patterns}

\theoremstyle{plain}
\newtheorem{theorem}{Theorem}
\newtheorem{proposition}{Proposition}
\newtheorem{lemma}{Lemma}
\newtheorem{corollary}{Corollary}
\newtheorem{definition}{Definition}
\newtheorem{rrule}{Reduction Rule}

\DeclareMathOperator{\nd}{nd}

\newcommand{\prob}[6]{%
  \begin{quote}
  	\begin{samepage}
    \begin{labeling}{#6}%
      \setlength\topsep{-.6ex} \setlength\itemsep{-.2ex}
    \item[#1]
    \item[\emph{#2}]#3
    \item[\emph{#4}]#5
    \end{labeling}%
	\end{samepage}
  \end{quote}%
}

\newcommand{\probdefdec}[3]{\prob{#1}{Input:}{#2}{Question:}{#3}{as}}

\newcommand{\probname}{\textsc{$c$-Closed Vertex Deletion}}
\newcommand{\probshort}{\textsc{$c$-CVD}}

\newcommand{\degree}{\ensuremath{\textsf{deg}}}

\newcommand{\Oh}{\mathcal{O}}

\title{A Complexity Analysis of the c-Closed~Vertex~Deletion~Problem}
\author{
Lisa Lehner, Christian Komusiewicz, Luca Pascal Staus\\\\
Friedrich Schiller University Jena, Jena, Germany\\
\texttt{\{c.komusiewicz,luca.staus\}@uni-jena.de}
}
\date{}

\begin{document}

\maketitle

\begin{abstract}
A graph is~$c$-closed when every pair of nonadjacent vertices has at most~$c-1$ common neighbors. In \probname, the input is a graph~$G$ and an integer~$k$ and we ask whether~$G$ can be transformed into a~$c$-closed graph by deleting at most~$k$ vertices. We study the classic and parameterized complexity of \probname. We obtain, for example, NP-hardness for the case that~$G$ is bipartite with bounded maximum degree. We also show upper and lower bounds on the size of problem kernels for the parameter~$k$ and introduce a new parameter, the number~$x$ of vertices in \emph{bad pairs}, for which we show a problem kernel of size~$\Oh(x^3 + x^2\cdot c)$. Here, a pair of nonadjacent vertices is \emph{bad} if they have at least~$c$ common neighbors. Finally, we show that \probname~can be solved in polynomial time on unit interval graphs with depth at most~$c+1$ and that it is fixed-parameter tractable with respect to the neighborhood diversity of~$G$.    
\end{abstract}

\section{Introduction}
One of the most striking features of social networks is that they tend
to show a high degree of \emph{triadic closure}. This is the
phenomenon that two agents which both interact with a common third
agent are likely to interact directly, too. There are different ways to
quantify this observation. For example, one can observe that
social networks have high clustering coefficients, meaning that a
large fraction of all connected vertex triples are triangles.
A further formal definition capturing triadic closure is graph~\emph{closure}, introduced by Fox et al.~\cite{FRSWW20}.
Here, a network is called~\emph{$c$-closed} if all pairs of vertices with at least~$c$ common neighbors are adjacent and the closure number of a network is the smallest number~$c$ such that the network is~$c$-closed. The hypothesis driven by the triadic closure phenomenon is now that networks should have a relatively small closure number and this is indeed the case for real-world data~\cite{FRSWW20,KNSS25}. Apart from giving further evidence for theories about the mechanisms that underly the creation of real-world networks, this observation has consequences for the development of algorithms for real-world networks: Many hard computational problems are fixed-parameter tractable when parameterized by~$c$ or by combinations of~$c$ and other parameters. Informally, this means that these problems can be solved (more) efficiently when the input networks have a small closure number.

For example, the enumeration of maximal cliques is fixed-parameter tractable with respect to the closure number~$c$~\cite{FRSWW20}. More precisely, all maximal cliques can be enumerated in~$3^{c/3}\cdot n^{\Oh(1)}$ time, where~$n$ denotes the number of vertices in the graph. This result has been extended to the enumeration of different types of near-cliques~\cite{FPT22,KKS20a}. Further algorithmic applications of $c$-closure were given for variants of \textsc{Dominating Set}~\cite{KMRSS23,KKS20,KKS20a},~\textsc{Independent Set}~\cite{KKS20,KKS20a}, and~\textsc{Vertex Cover}~\cite{KKNS22,KKS23} and for problems related to finding certain small subgraphs~\cite{KN21}.

A drawback of $c$-closure is that it is defined as the \emph{maximum} of the sizes of common neighborhoods of nonadjacent vertices (plus one) which makes it sensitive to local changes: all it takes to drastically increase the $c$-closure of a graph is the absence of one edge between two vertices with many common neighbors. Ideally, we would like to obtain more robust positive algorithmic results that hold not only for~$c$-closed graphs but also for those graphs which are \emph{almost}~$c$-closed. 

We follow the natural way of quantifying this by counting the number~$k$ of vertex deletions that are necessary to obtain a $c$-closed graph. We show that positive results for \textsc{Independent Set} and \textsc{Clique} can be easily extended from $c$-closed graphs to almost $c$-closed graphs. 
\begin{proposition}
  Let~$G=(V,E)$ be a graph and~$S$ a size-$k$ vertex set in~$G$ such that~$G-S$ is~$c$-closed. Then, we can
 (1) determine in~$2^k\cdot f(\ell+c)\cdot n^{\Oh(1)}$ time whether~$G$ contains an independent set of size at least~$\ell$, and
 (2) enumerate all maximal cliques of~$G$ in~$2^k\cdot (3^{c/3})\cdot n^{\Oh(1)}$ time.
  \label{prop:application}
\end{proposition}
\begin{proof}
  To determine whether~$G$ contains an independent set~$I$ on~$\ell$ vertices, we use the following known algorithm which can be formulated for vertex deletion sets~$S$ to any hereditary graph class for which independent sets can be found efficiently. Branch into all cases for~$S'\coloneqq S\cap I$. In each branch,  first check whether~$S'$ is an independent set. If this is the case, then delete all vertices from~$G-S$ that are adjacent to~$S'$ (as they cannot be contained in an independent set that includes~$S'$) and determine whether the resulting graph has an independent set of size~$\ell-|S'|$. Since the resulting graph is $c$-closed, this can be done in~$f(\ell+c)\cdot n^{\Oh(1)}$ time~\cite{KKS20}. Since~$S$ has size~$k$, we consider altogether~$2^k$ cases for~$S'$ which gives the claimed running time.  

To enumerate all maximal cliques of~$G$ proceed as follows. For each subset~$S'$ of~$S$, enumerate all maximal cliques~$K$ such that~$K\cap S = S'$. To do this, delete from~$G-S$ all vertices that are not adjacent to all vertices of~$S'$. Then enumerate all maximal cliques in the resulting graph. Since this graph is~$c$-closed, this enumeration can be done in~$3^{c/3}\cdot n^{\Oh(1)}$ time. For each enumerated clique~$K'$ we check in polynomial time whether~$S'\cup K'$ is a maximal clique in~$G$ and output~$S'\cup K'$ if this is the case. The running time bound now follows from the fact that~$2^k$ cases for~$S'$ are considered. The correctness follows from the fact that only maximal cliques are output and that each maximal clique~$K$ is output when the branch~$S'=K\cap S$ is considered.
\end{proof}

Proposition~\ref{prop:application} implies that it may be worthwhile to find out whether a graph~$G$ is almost~$c$-closed and to identify the respective vertex deletion set~$S$ in that case.
Motivated by this, we investigate  the complexity of the corresponding decision problem.
\probdefdec
{\probname\ (\probshort)}
{An undirected graph~$G = (V,E)$ and a positive integer~$k$.}
{Is there a set~$S \subseteq V$ with~$|S| \leq k$ such that~$G-S$ is~$c$-closed?}
Unless mentioned otherwise, throughout the paper we assume that $c$ is constant.
\paragraph{Known Results and Further Related Work.}
Since  being $c$-closed is hereditary, \probshort~is NP-hard for every fixed value of~$c$~\cite{LY80}. Assuming the exponential time hypothesis (ETH), one cannot even solve \probshort~in~$2^{o(n+m)}$ time, where~$m$ is the number of edges in the graph~\cite{K18}.
The case~$c=1$ is equivalent to \textsc{Cluster Vertex Deletion} since a graph is~1-closed if and only if it does not contain a~$P_3$, a path on three vertices, as induced subgraph.  \textsc{Cluster Vertex Deletion} is NP-hard even when restricted to graphs with maximum degree~$3$~\cite{Rusu25}. This gives the following complexity dichotomy: For maximum degree~$\Delta\le 2$, \textsc{1-CVD} is easily solvable in linear time, all other cases are NP-hard.

A further related parameter is the \emph{weak} closure number~$\gamma$.  This is the smallest number $\gamma$ such that every induced subgraph of~$G$ contains \emph{some} vertex~$v$ that has at most~$\gamma$~common~neighbors with each of its nonneighbors. The value of~$\gamma$ is upper-bounded by the closure number~$c$ and by the degeneracy of the input graph. As a consequence, it assumes very small values in real-world data~\cite{FRSWW20,KNSS25}. There are several problems where a small weak closure number~$\gamma$ can be exploited, for example \textsc{Clique Enumeration}~\cite{FRSWW20}, \textsc{Dominating Set}~\cite{LS21}, \textsc{Independent Set}~\cite{KKS20a}, \textsc{Induced Matching}, \textsc{Capacitated Vertex Cover} and \textsc{Connected Vertex Cover}~\cite{KKS23}. As we show in Section~\ref{sec:prelims}, the weak closure number~$\gamma$ of a graph is unrelated to~$k+c$ where~$k$ is defined as the vertex deletion distance to~$c$-closed graphs. More precisely, we show that there are graphs with unbounded~$\gamma$-value where~$k+c$ is constant and that there are graphs where~$\gamma$ is constant but~$k+c$ is unbounded.

\paragraph{Our Results.}

\begin{figure}[t]
\center
\scalebox{0.65}{
\begin{tikzpicture}

\draw[draw=none,fill=gray!10] (0.5,-0.5) rectangle ++ (6,-6); 
\draw[draw=none,fill=green!50] (0.5,-0.5) rectangle ++ (2,-1.5); 
\draw[draw=none,fill=red!50] (2.5,-0.5) rectangle ++ (4,-1); 
\draw[draw=none,fill=green!50] (0.5,-1.5) rectangle ++ (1,-5); 
\draw[draw=none,fill=green!50] (1.5,-1.5) rectangle ++ (1,-5); 
\draw[draw=none,fill=green!50] (2.5,-2.5) rectangle ++ (1,-4); 
\draw[draw=none,fill=green!50] (3.5,-4.5) rectangle ++ (1,-2); 
\draw[draw=none,fill=green!50] (4.5,-5.5) rectangle ++ (1,-1); 

\draw[draw=none,fill=red!50] (3.5,-1.5) rectangle ++ (1,-1); 
\draw[draw=none,fill=red!50] (4.5,-1.5) rectangle ++ (1,-2); 
\draw[draw=none,fill=red!50] (5.5,-1.5) rectangle ++ (1,-5); 

\draw [line width=1.5pt] (-0.5,-0.5) -- (6.5,-0.5);
\draw [line width=1.5pt] (0.5,0.5) -- (0.5,-6.5);
\draw [line width=1.5pt] (-0.4,0.4) -- (0.5,-0.5);

\node[label={\Large $\Delta$}] at (0.2, -0.2){};
\node[label={\Large $c$}] at (-0.2, -0.6){};

\node[label={\Large $1$}] at (1, -0.4){};
\node[label={\Large $2$}] at (2, -0.4){};
\node[label={\Large $3$}] at (3, -0.4){};
\node[label={\Large $4$}] at (4, -0.4){};
\node[label={\Large $5$}] at (5, -0.4){};
\node[label={\Large $6$}] at (6, -0.4){};

\node[label={\Large $1$}] at (0, -1.4){};
\node[label={\Large $2$}] at (0, -2.4){};
\node[label={\Large $3$}] at (0, -3.4){};
\node[label={\Large $4$}] at (0, -4.4){};
\node[label={\Large $5$}] at (0, -5.4){};
\node[label={\Large $6$}] at (0, -6.4){};

\node[label={\Large NP-hard}] at (5, -2.4){};
\node[label={\Large P}] at (2, -4.4){};

\end{tikzpicture}
}
\caption{Hardness for different values of~$c$ and~$\Delta$. Red means NP-hard, green means polynomial-time solvable, and gray means unknown.}
\label{fig:max_deg_hardness_table}
\end{figure}
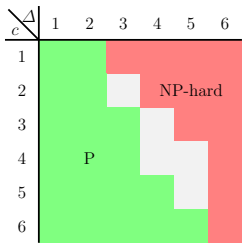
In Section~\ref{sec:degree}, we make progress towards a maximum degree-based complexity dichotomy for~\probshort; the results are shown in Figure~\ref{fig:max_deg_hardness_table}. In particular, our results show that for all~$c\ge 6$, \probshort~is NP-hard when restricted to graphs with maximum degree~$\Delta\ge c$ and polynomial-time solvable when restricted to graphs with maximum degree~$\Delta\le c-1$. We also show that for~$c=\Delta=3$, the problem is polynomial-time solvable so the complexity landscape seems more complicated for smaller values of~$c$. Our results also show that the problem is NP-hard on bipartite graphs.

In Section~\ref{sec:kernel}, we then consider the parameterized complexity of the problem with respect to the natural parameter~$k$. We show that a kernel with~$\Oh(k^{c+2})$~vertices can be computed in time~$\Oh(n^3 + n^2 m)$ and also provide a lower bound of~$\Oh(k^{c-\epsilon})$ for the bit size of the kernel.
To obtain kernels whose size depends only polynomially on~$c$, we introduce a further natural parameter, the number~$x$ of \emph{vertices in bad pairs} in~$G$, that is, the total number of vertices that have at least~$c$ common neighbors with some nonneighbor. This parameter can be thought of as being larger than~$k$ since deleting all vertices in a bad pair makes the graph~$c$-closed. We show that \probshort{} admits a kernel with $\Oh(x^3 + x^2\cdot c)$~vertices.

Finally, in Section~\ref{sec:easy} we show that \probshort{} is fixed-parameter tractable with respect to the neighborhood diversity and that it can be solved in polynomial time on unit interval graphs with maximum clique size~$c+1$.

\section{Preliminaries}
\label{sec:prelims}
An \emph{undirected graph}~$G=(V,E)$ consists of a \emph{vertex set}~$V$ and an \emph{edge set}~$E\subseteq \{ \{ u,v \} \mid u,v\in V \wedge u \neq v \}$.
We let~$V(G)$ and~$E(G)$ denote the vertices and edges of~$G$, respectively, and define~$n \coloneq |V(G)|$ and~$m\coloneq |E(G)|$.
Two vertices~$u$ and~$v$ are \emph{neighbors} if~$\{u,v\}\in E(G)$.
The \emph{neighborhood} of a vertex~$v$ in~$G$ is~$N_G(v) \coloneq \{ u\in V(G) \mid \{ u,v \}\in E(G) \}$.
The \emph{degree} of~$v$ is~$\degree_G(v) \coloneq |N_G(v)|$.
The \emph{maximum degree} of~$G$ is~$\Delta_G \coloneq \max_{v\in V(G)} \degree(v)$.
We drop the subscript~$G$ if~$G$ is clear from context.
A graph~$G'$ is a \emph{subgraph} of~$G$ if~$V(G')\subseteq V(G)$ and~$E(G')\subseteq E(G)$.
Let~$U\subseteq V(G)$.
The graph~$G[U] \coloneq (U, \{ e\in E(G) \mid e\subseteq U \} )$ is the \emph{subgraph of~$G$ induced by~$U$}.
We write~$G - U \coloneq G[V(G)\setminus U]$ to denote the subgraph of~$G$ obtained by deleting the vertices in~$U$.

Let~$G$ be a graph.
We call~$G$ a \emph{clique} if it contains every possible edge.
We call~$G$ an \emph{independent set} if it contains no edges.
We call~$G$ \emph{bipartite} if~$V(G)$ can be partitioned into two sets~$V_1$ and~$V_2$ such that~$G[V_1]$ and~$G[V_2]$ are independent sets.
We call~$G$ a \emph{split graph} if~$V(G)$ can be partitioned into two sets~$V_1$ and~$V_2$ such that~$G[V_1]$ is a clique and~$G[V_2]$ is an independent set.

A graph property~$\Pi$ is \emph{hereditary} if it is \emph{closed under vertex deletion}. In other words, any induced subgraph of a graph with property~$\Pi$ also has property~$\Pi$. All graph properties defined above, including~$c$-closed graphs, are hereditary.~\cite{FRSWW20}

\begin{figure}[t]
    \centering
        \begin{tikzpicture}[scale=0.55, every node/.style={circle, fill, inner sep=2pt, minimum size=0pt, draw=none}]
        
			\newcommand{\shiftA}{6}
			\newcommand{\shiftB}{8 + \shiftA}
        
			\node[white, label=below:{\normalsize $c = 2$}] (1) at (1, 1.5){};
        	
            \node (WA) at (0,2) [circle,draw] {};
            \node (XA) at (2,2) [circle,draw] {};
            \node[fill=lightgray] (UA) at (0,4) [circle,draw] {};
            \node[fill=lightgray] (VA) at (2,4) [circle,draw] {};

            \draw (UA) -- (WA);
            \draw (UA) -- (XA);
            \draw (VA) -- (WA);
            \draw (VA) -- (XA);
            \draw[dashed] (WA) -- (XA);

			\node[white, label=below:{\normalsize $c = 3$}] (1) at (2 + \shiftA, 1.5){};
            
            \node (WB) at (0 + \shiftA,2) [circle,draw] {};
            \node (XB) at (2 + \shiftA,2) [circle,draw] {};
            \node (YB) at (4 + \shiftA,2) [circle,draw] {};
            \node[fill=lightgray] (UB) at (1 + \shiftA,4) [circle,draw] {};
            \node[fill=lightgray] (VB) at (3 + \shiftA,4) [circle,draw] {};
            \draw (UB) -- (WB);
            \draw (UB) -- (XB);
            \draw (UB) -- (YB);
            \draw (VB) -- (WB);
            \draw (VB) -- (XB);
            \draw (VB) -- (YB);
            \draw[dashed] (WB) -- (XB);
            \draw[dashed] (XB) -- (YB);
            \draw[dashed] (WB) to[out=-45, in=-135] (YB);

			\node[white, label=below:{\normalsize $c = 4$}] (1) at (2.25 + \shiftB, 1.5){};
            
            \node (WC) at (0 + \shiftB,2) [circle,draw] {};
            \node (XC) at (1.5 + \shiftB,2) [circle,draw] {};
            \node (YC) at (3 + \shiftB,2) [circle,draw] {};
            \node (AC) at (4.5 + \shiftB,2) [circle,draw] {};
            \node[fill=lightgray] (UC) at (1 + \shiftB,4) [circle,draw] {};
            \node[fill=lightgray] (VC) at (3.5 + \shiftB,4) [circle,draw] {};

            \draw (UC) -- (WC);
            \draw (UC) -- (XC);
            \draw (UC) -- (YC);
            \draw (VC) -- (WC);
            \draw (VC) -- (XC);
            \draw (VC) -- (YC);
            \draw (VC) -- (AC);
            \draw (UC) -- (AC);
            \draw[dashed] (WC) -- (XC);
            \draw[dashed] (XC) -- (YC);
            \draw[dashed] (WC) to[out=-45, in=-135] (YC);
            \draw[dashed] (AC) to[out=-45, in=-135] (WC);
            \draw[dashed] (AC) -- (YC);
            \draw[dashed] (AC) to[out=-45, in=-135] (XC);

        \end{tikzpicture}
    \caption{Minimal forbidden subgraphs (FSG) for~$c = 2,3,4$. The grey vertices are the bad pair, the black vertices are the connecting vertices, the full edges are the critical edges, and the dashed edges are optional edges.}
    \label{fig:fsgs}
\end{figure}
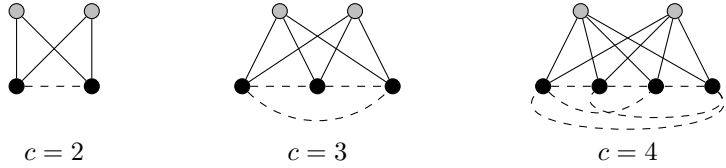

\paragraph{Problem-Specific Definitions.}

    A graph~$G = (V,E)$ is \emph{$c$-closed} for some positive integer~$c$ if every pair of nonadjacent vertices has at most~$c - 1$ common neighbors. The \emph{closure number} of $G$ is the smallest number~$c$ such that~$G$ is~$c$-closed.
We call a pair of nonadjacent vertices with at least~$c$ common neighbors a \emph{bad pair}.
The common neighbors of a bad pair are its \emph{connecting vertices}.
Clearly, a graph is~$c$-closed if and only if it contains no bad pair. In other words, the~$c$-closed property can be defined based on the following forbidden subgraphs:
A \emph{minimal forbidden subgraph} (FSG) for being~$c$-closed is a graph with~$c + 2$ vertices such that two of the vertices form a bad pair and the remaining~$c$ vertices are their connecting vertices. The adjacency between the connecting vertices can be arbitrary.
Given an FSG~$G'$, we refer to the edges between the bad pair and the connecting vertices as \emph{critical edges} of~$G'$. We call an edge~$e$ \emph{critical} in~$G$ if~$G$ contains an FSG where~$e$ is a critical edge.
Figure~\ref{fig:fsgs} shows the FSGs for~$c=2,3,4$.
Note that a bad pair can have more than~$c$ connecting vertices if it is the bad pair in multiple FSGs. Moreover, a minimal non-$c$-closed subgraph can have several bad pairs; for convenience, we let FSGs often refer to a combination of a subgraph and a bad pair.
\paragraph{Relation to Weak Closure.}    
The parameter combination~$k + c$ is unrelated to the closure number and the weak closure number~$\gamma$ of a graph for any fixed~$c$. Note that in our setting~$c$ is not the closure number of the graph but instead the desired closure number after deleting up to~$k$ vertices. We first show that~$k + c$ can be arbitrarily smaller than~$\gamma$. Let~$s$ be any arbitrarily large integer and let~$G$ be a graph containing a clique with~$2s$~vertices and two additional vertices~$v_1$ and~$v_2$. The vertex~$v_1$ is connected to~$s$ of the clique vertices and~$v_2$ is connected to the other~$s$~clique vertices. Clearly, we can make this graph~$1$-closed by removing~$v_1$ and~$v_2$. This means we have~$k + c \leq 2 + c$ for any fixed~$c\ge 1$. In contrast, we have~$\gamma\ge s + 1$ since each vertex has at least one nonneighbor with~$s$ common neighbors. 

Conversely,~$k + c$ can also be arbitrarily larger than the closure number of a graph: Let~$G$ be a graph with an independent set of size~$c$ and two vertices~$v_1$ and~$v_2$ that are adjacent to each vertex in the independent set. A graph consisting of~$s\gg c + 1$ copies of this component is~$c + 1$-closed. To make it~$c$-closed we need to remove one vertex from each of the~$s$ components.

\paragraph{Parameterized Complexity.}
A \emph{parameterized problem} is a problem where each instance consists of the problem input~$\mathcal{I}$ and a problem specific parameter~$k$.
A parameterized problem is \emph{fixed-parameter tractable (FPT)} for the parameter~$k$ if an algorithm with running time~$f(k) \cdot |\mathcal{I}|^{\Oh(1)}$ exists for some computable function~$f: \mathbb{N} \rightarrow \mathbb{N}$.
A \emph{data reduction rule} is an algorithm that runs in polynomial time and transforms any problem instance~$(\mathcal{I},k)$ of a parameterized problem into an equivalent instance~$(\mathcal{I}',k')$ of the same problem.
A data reduction rule is called a \emph{kernelization algorithm} or just \emph{kernel} for the parameter~$k$ if the size of any instance~$(\mathcal{I}',k')$ that the data reduction rule produces is bounded by a function that only depends on~$k$, that is~$\mathcal{I}' + k' \leq g(k)$ for some computable function~$g: \mathbb{N} \rightarrow \mathbb{N}$.
For a more complete introduction to parameterized complexity we refer to~\cite{PA16}.

\section{Bounded Maximum Degree~$\mathbf{\Delta}$}
\label{sec:degree}
We now study the complexity of~$\probshort$ for different values of~$c$ and~$\Delta$. Figure~\ref{fig:max_deg_hardness_table} shows an overview of these results.
    
\subsection{NP-Hard Cases}
We reduce from \textsc{Vertex Cover} which is NP-hard when the maximum degree is~$3$ and all degree-$3$ vertices have distance at least~$3$ to each other~\cite{K18}.

\begin{figure}[t]
        \centering
        \begin{subfigure}[t]{0.49\textwidth}
            \centering
            \scalebox{0.8}{
            \begin{tikzpicture}[scale=1, every node/.style={circle, fill, inner sep=2pt, minimum size=0pt, draw=none}]
                \node (A) at (0,0) [circle,draw] {};
                \node (B) at (4,0) [circle,draw] {};
                \node (C) at (6,2) [circle,draw] {};
                \node (D) at (2,2) [circle,draw] {};
                \draw (A) -- (B);
                \draw (A) -- (D);
                \draw (B) -- (C);
                \draw (B) -- (D);
            \end{tikzpicture}
            }  
            \caption{\textsc{Vertex Cover} instance with~$\Delta = 3$.}
            \label{fig:maxdeg6a}
        \end{subfigure}
        \begin{subfigure}[t]{0.49\textwidth}
            \centering
            \scalebox{0.8}{
            \begin{tikzpicture}[scale=1, every node/.style={circle, fill, inner sep=2pt, minimum size=0pt, draw=none}]
                \node (A) at (0,0) [circle,draw] {};
                \node (B) at (4,0) [circle,draw] {};
                \node (C) at (6,2) [circle,draw] {};
                \node (D) at (2,2) [circle,draw] {};
                \node[fill=white] (AD1) at (0.5,1.3) [circle,draw] {};
                \node[fill=gray] (AD2) at (1.2,1.2) [circle,draw] {};
                \node[fill=white] (AD3) at (1.3,0.55) [circle,draw] {};

                \node[fill=white] (AB1) at (2,0.4) [circle,draw] {};
                \node[fill=gray] (AB2) at (2.3,0) [circle,draw] {};
                \node[fill=white] (AB3) at (2,-0.4) [circle,draw] {};

                \node[fill=white] (CB1) at (4.5,1.3) [circle,draw] {};
                \node[fill=gray] (CB2) at (5.2,1.2) [circle,draw] {};
                \node[fill=white] (CB3) at (5.3,0.55) [circle,draw] {};

                \node[fill=white] (DB1) at (3.3,1.3) [circle,draw] {};
                \node[fill=gray] (DB2) at (3.2,0.625) [circle,draw] {};
                \node[fill=white] (DB3) at (2.5,0.55) [circle,draw] {};

                \node[fill=gray] (AB4) at (1.7,0) [circle,draw] {};
                \node[fill=gray] (DB4) at (2.65,1.2) [circle,draw] {};
                \node[fill=gray] (CB4) at (4.65,0.65)  [circle,draw] {};
                \node[fill=gray] (AD4) at (0.65,0.65) [circle,draw] {};

                \draw (A) -- (AB1);
                \draw (AB1) -- (AB2);
                \draw (A) -- (AB3); 
                \draw (B) -- (AB1);
                \draw (AB3) -- (AB2);
                \draw (B) -- (AB3);
                \draw (A) -- (AD1);
                \draw (AD1) -- (AD2);
                \draw (A) -- (AD3); 
                \draw (D) -- (AD1);
                \draw (AD3) -- (AD2);
                \draw (D) -- (AD3);   
                \draw (D) -- (DB1);
                \draw (DB1) -- (DB2);
                \draw (D) -- (DB3); 
                \draw (B) -- (DB1);
                \draw (DB3) -- (DB2);
                \draw (B) -- (DB3); 
                \draw (C) -- (CB1);
                \draw (CB1) -- (CB2);
                \draw (C) -- (CB3); 
                \draw (B) -- (CB1);
                \draw (CB3) -- (CB2);
                \draw (B) -- (CB3); 
                \draw (AB1) -- (AB4);
                \draw (AB3) -- (AB4); 
                \draw (DB3) -- (DB4); 
                \draw (DB1) -- (DB4); 
                \draw (AD1) -- (AD4); 
                \draw (AD3) -- (AD4); 
                \draw (CB3) -- (CB4);
                \draw (CB1) -- (CB4);  
            \end{tikzpicture}
            }
            \caption{\textsc{$4$-CVD instance with~$\Delta = 6$.}}
            \label{fig:maxdeg6b}
        \end{subfigure}

        \caption{Construction of a \textsc{$4$-CVD} instance with maximum degree~$6$ from a \textsc{Vertex Cover} instance with maximum degree~$3$. The vertices from the \textsc{Vertex Cover} instance are black, the bad pair vertices are white and the additional connecting vertices are gray.}
        \label{fig:maxdeg6}
\end{figure}
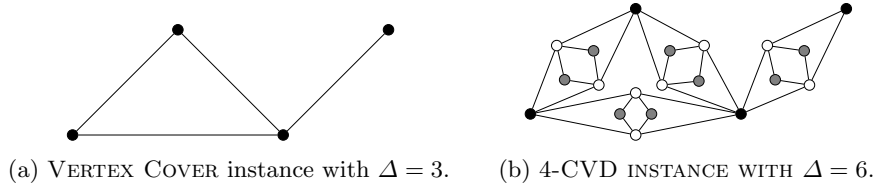

\begin{theorem}\label{lem:maxdeg6}
	$\probname$ is NP-hard for~$c \geq 2$ on bipartite graphs with maximum degree~$\Delta = \max(c, 6)$.
\end{theorem}
\begin{proof}\label{proof:maxdeg6}
	We reduce from a \textsc{Vertex Cover} instance~$(G,k)$ with maximum degree~$3$. The general idea is to replace each edge with an FSG. We create a new graph~$G'$ as follows: Each vertex~$v\in V(G)$ becomes a vertex in~$G'$. Additionally, for each edge~$e\in E(G)$ we add two bad pair vertices~$b_1^e, b_2^e$ and~$c - 2$ connecting vertices~$w_1^e, \hdots, w_{c - 2}^e$. We add edges from the bad pair vertices~$b_1^e, b_2^e$ to all connecting vertices and to the two endpoints of~$e$. Figure~\ref{fig:maxdeg6} shows this construction. Clearly, the vertices from the original graph~$G$ have degree at most~$6$ since they are now connected to two bad pair vertices per edge. The additional connecting vertices always have degree~$2$ and the bad pair vertices always have degree~$c$. Hence, the whole graph has maximum degree~$\max(c,6)$. It is also easy to see that~$G'$ is bipartite. The bad pair vertices that are added for each edge in the original graph form one side and the remaining vertices form the other side. This finishes the construction of the~$\probshort$ instance~$(G',k)$. We now need to show the following:
\begin{align*}
	(G,k) \text{ is a yes-instance} \Leftrightarrow (G',k) \text{ is a yes-instance.}
\end{align*}

	($\Rightarrow$) Let~$S$ be a solution for~$(G,k)$. For each~$e\in E(G)$, at least one of the~$c$ neighbors of the vertices~$b_1^e$ and~$b_2^e$ is in~$S$ which means they cannot form a bad pair in~$G' - S$. The vertices~$w_1^e, \hdots, w_{c - 2}^e$ have degree~$2$ but only exist if~$c > 2$. These can also not be part of a bad pair. Each connected component in~$G' - S$ can only contain at most one vertex from~$V(G)$ since~$S$ contains at least one endpoint of each edge. Together this means that~$G' - S$ cannot contain an FSG.
	
	($\Leftarrow$) Let~$S'$ be a solution for~$(G',k)$. We now construct a set~$S$ that is a solution for~$(G,k)$. First, we add all vertices in~$S' \cap V(G)$ to~$S$. Each vertex~$v\in S' \setminus V(G)$ is either a bad pair vertex~$b_i^e$ or a connecting vertex~$w_i^e$ for some edge~$e\in E(G)$. In either case we add an arbitrary endpoint of~$e$ to~$S$. Clearly,~$S$ contains at most~$k$ vertices. For each edge~$e\in E(G)$,~$S'$ contains at least one vertex from the corresponding FSG. Therefore~$S$ must contain at least one vertex from each edge in~$E(G)$.
\end{proof}

The following theorem proves further hard cases for~$c=2$ and~$c=3$.

\begin{figure}[t]
        \centering
        \begin{subfigure}[t]{0.49\textwidth}
            \centering
            \scalebox{0.8}{
            \begin{tikzpicture}[scale=1]
            	\node[fill,inner sep=2pt,minimum size=0pt] (A) at (0,0) [circle,draw] {};
            	\node at (A) [above,yshift=1mm] {\large $v$};
                \node[fill,inner sep=2pt,minimum size=0pt] (B) at (-2,-2) [circle,draw] {};
            	\node at (B) [left,xshift=-1mm] {\large $u_1$};
                \node[fill,inner sep=2pt,minimum size=0pt] (C) at (0,-2.5) [circle,draw] {};
            	\node at (C) [left,xshift=-1mm] {\large $u_2$};
                \node[fill,inner sep=2pt,minimum size=0pt] (D) at (2,-2) [circle,draw] {};
            	\node at (D) [right,xshift=1mm] {\large $u_3$};
                \node[fill,inner sep=2pt,minimum size=0pt] (E) at (-2,-3) [circle,draw] {};
                \node[fill,inner sep=2pt,minimum size=0pt] (F) at (0,-3.5) [circle,draw] {};
                \node[fill,inner sep=2pt,minimum size=0pt] (G) at (2,-3) [circle,draw] {};
                \draw (A) -- (B);
                \draw (A) -- (C);
                \draw (A) -- (D);
                \draw (B) -- (E);
                \draw (C) -- (F);
                \draw (D) -- (G);
            \end{tikzpicture}
            }
            \caption{\textsc{Vertex Cover} instance with~$\Delta = 3$.}
            \label{fig:maxdeg45a}
        \end{subfigure}
        \begin{subfigure}[t]{0.49\textwidth}
            \centering
            \scalebox{0.8}{
            \begin{tikzpicture}[scale=1]
            	
            	\node[fill,inner sep=2pt,minimum size=0pt] (A) at (0,0) [circle,draw] {};
            	\node at (A) [above,yshift=1mm] {\large $v$};
                \node[fill,inner sep=2pt,minimum size=0pt] (B) at (-2,-2) [circle,draw] {};
            	\node at (B) [left,xshift=-1mm] {\large $u_1$};
                \node[fill,inner sep=2pt,minimum size=0pt] (C) at (0,-2.5) [circle,draw] {};
            	\node at (C) [left,xshift=-1mm] {\large $u_2$};
                \node[fill,inner sep=2pt,minimum size=0pt] (D) at (2,-2) [circle,draw] {};
            	\node at (D) [right,xshift=1mm] {\large $u_3$};
                \node[fill,inner sep=2pt,minimum size=0pt] (E) at (-2,-3) [circle,draw] {};
                \node[fill,inner sep=2pt,minimum size=0pt] (F) at (0,-3.5) [circle,draw] {};
                \node[fill,inner sep=2pt,minimum size=0pt] (G) at (2,-3) [circle,draw] {};
                \node[fill=gray,inner sep=2pt,minimum size=0pt] (AB) at (-1,-1) [rectangle,draw] {};
            	\node at (AB) [below right, xshift=-1mm] {\large $s_1^v$};
                \node[fill=gray,inner sep=2pt,minimum size=0pt] (AC) at (0,-1.25) [rectangle,draw] {};
            	\node at (AC) [below left] {\large $s_2^v$};
                \node[fill=gray,inner sep=2pt,minimum size=0pt] (AD) at (1,-1) [rectangle,draw] {};
            	\node at (AD) [below left] {\large $s_3^v$};
                \node[fill=white,inner sep=2pt,minimum size=0pt] (BE1) at (-2.5,-2.5) [circle,draw] {};
                \node[fill=gray,inner sep=2pt,minimum size=0pt] (BE2) at (-2,-2.5) [circle,draw] {};
                \node[fill=white,inner sep=2pt,minimum size=0pt] (BE3) at (-1.5,-2.5) [circle,draw] {};
                \node[fill=white,inner sep=2pt,minimum size=0pt] (CF1) at (-0.5,-3) [circle,draw] {};
                \node[fill=gray,inner sep=2pt,minimum size=0pt] (CF2) at (0,-3) [circle,draw] {};
                \node[fill=white,inner sep=2pt,minimum size=0pt] (CF3) at (0.5,-3) [circle,draw] {};
                \node[fill=white,inner sep=2pt,minimum size=0pt] (DG1) at (1.5,-2.5) [circle,draw] {};
                \node[fill=gray,inner sep=2pt,minimum size=0pt] (DG2) at (2,-2.5) [circle,draw] {};
                \node[fill=white,inner sep=2pt,minimum size=0pt] (DG3) at (2.5,-2.5) [circle,draw] {};
                
                \node[fill=gray,inner sep=2pt,minimum size=0pt] (X1) at (-1.5,-0.5) [rectangle,draw] {};
            	\node at (X1) [left,xshift=-1mm] {\large $x_1^v$};
                \node[fill=gray,inner sep=2pt,minimum size=0pt] (X2) at (1.5,-0.5) [rectangle,draw] {};
            	\node at (X2) [right,xshift=1mm] {\large $x_2^v$};
                
                \draw (X1) to[out=-100, in=90] (B);
                \draw (X1) to[out=-90, in=135] (C);
                \draw (X1) to[out=-80, in=180] (D);
                \draw (X2) to[out=-100, in=0] (B);
                \draw (X2) to[out=-90, in=45] (C);
                \draw (X2) to[out=-80, in=90] (D);
                \draw (X1) -- (A);
                \draw (X2) -- (A);
                \draw (X1) -- (X2);

                \draw (A) -- (AB);
                \draw (AB) -- (B);
                \draw (A) -- (AC);
                \draw (AC) -- (C);
                \draw (A) -- (AD);
                \draw (AD) -- (D);
                
                \draw (B) -- (BE1);
                \draw (B) -- (BE3);
                \draw (BE1) -- (E);
                \draw (BE3) -- (E);
                \draw (BE1) -- (BE2);
                \draw (BE3) -- (BE2);
                \draw (C) -- (CF1);
                \draw (C) -- (CF3);
                \draw (CF1) -- (F);
                \draw (CF3) -- (F);
                \draw (CF1) -- (CF2);
                \draw (CF3) -- (CF2);
                \draw (D) -- (DG1);
                \draw (D) -- (DG3);
                \draw (DG1) -- (G);
                \draw (DG3) -- (G);
                \draw (DG1) -- (DG2);
                \draw (DG3) -- (DG2);

            \end{tikzpicture}
            }
            \caption{\textsc{$3$-CVD instance with~$\Delta = 5$.}}
            \label{fig:maxdeg45b}
        \end{subfigure}

        \caption{Construction of a \textsc{$3$-CVD} instance with maximum degree~$5$ from a \textsc{Vertex Cover} instance with maximum degree~$3$. The vertices from the \textsc{Vertex Cover} instance are black. The gray square vertices are the connecting vertices for the FSGs corresponding to the edges incident with the degree-$3$ vertex. The white circle vertices are the bad pair vertices for the FSGs corresponding to other edges with the gray circle vertices as additional connecting vertices.}
        \label{fig:maxdeg45}
\end{figure}
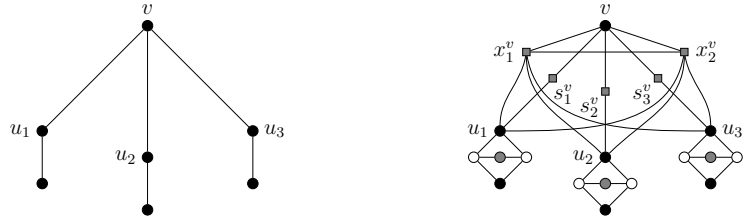

\begin{theorem}\label{lem:maxdeg45}
  $\probname$ is NP-hard for~$c = 2$ and~$\Delta = 4$ and for~$c = 3$ and~$\Delta = 5$.
\end{theorem}
\begin{proof}
	Let~$(G,k)$ be a \textsc{Vertex Cover} instance with maximum degree~$3$ where all degree-$3$ vertices have distance at least~$3$ to each other. We want to replace each edge with an FSG but we will use a special construction for edges that have a degree-$3$ vertex as an endpoint. We construct a new graph~$G'$ as follows: Each vertex~$v\in V(G)$ becomes a vertex in~$G'$. For each edge~$e\in E(G)$ where both endpoints have degree at most~$2$, we add two bad pair vertices~$b_1^e, b_2^e$ and~$c - 2$ connecting vertices~$w_1^e, \hdots, w_{c - 2}^e$. We add edges from the bad pair vertices to all connecting vertices and the two endpoints of~$e$. This is the same construction we used in the proof of Theorem~\ref{lem:maxdeg6}.
	With this,~$G'$ currently has a maximum degree of~$c$ while the vertices in~$V(G')\cap V(G)$ all have at most degree~$2$.
For the remaining edges we look at each degree~$3$ vertex individually. Let~$v\in V(G)$ be a degree~$3$ vertex in~$G$ and let~$u_1, u_2, u_3$ be the three neighbors of~$v$ in~$G$. We now add the vertices~$s_1^v, s_2^v, s_3^v$ to~$G'$ and connect them to~$v$. Additionally, we connect~$u_i$ to~$s_i^v$ for each~$i\in \{1,2,3\}$.
If~$c = 2$, then we add one additional vertex~$x_1^v$ and connect it to~$v,u_1, u_2, u_3$. If instead~$c = 3$, then we add two vertices~$x_1^v, x_2^v$ and connect them to~$v,u_1, u_2, u_3$ and each other. The construction is visualized in Figure~\ref{fig:maxdeg45}. The~$s_i^v$ vertices have degree~$2$ while~$v$, the~$x_i^v$ vertices, and the~$u_i$ vertices have degree~$4$ and~$5$ for~$c = 2$ and~$c = 3$, respectively. That means~$G'$ has a maximum degree of~$4$ for~$c = 2$ and~$5$ for~$c = 3$. This finishes the construction of the~$\probshort$ instance~$(G',k)$.

	We will first show that for every FSG~$H$ in~$G'$ we have~$V(H)\cap V(G) = e$ for some edge~$e\in E(G)$. Let~$H$ be some FSG in~$G'$. First, observe that~$H$ cannot contain three vertices from~$V(G)$: The vertices in~$V(G)$ are never adjacent in~$G'$. The FSG~$H$ can therefore only contain three of them if they are all connecting vertices and~$c = 3$. The only vertices that have at least three neighbors from~$V(G)$ are the~$x_1^v, x_2^v$ vertices. These two must therefore form the bad pair which is not possible because they are adjacent.
Next, we show that~$H$ cannot contain two vertices from~$V(G)$ that are not adjacent in~$G$: The vertices in~$H$ must have at most distance~$2$ to each other. The only vertices in~$V(G)$ that are not connected in~$G$ for which this is the case are the neighbors~$u_1,u_2,u_3$ of some degree-$3$ vertex~$v$ in~$G$. However, for~$c = 3$ they only have~$x_1^v$ and~$x_2^v$ as common neighbors and those two vertices cannot form a bad pair since they are adjacent. For~$c = 2$ they only have~$x_1^v$ as a common neighbor which also cannot form a bad pair on its own.
Finally, we show that~$H$ has to contain at least two vertices from~$V(G)$: Any FSG must contain a cycle of length~$4$ as a subgraph. That means if~$H$ contains at most one vertex from~$V(G)$, then~$H$ must contain a path on three vertices from~$V(G')\setminus V(G)$ as a subgraph. In~$G'$ the only instances of such subgraphs are the vertices~$b_1^e, b_2^e, w_1^e$ for some edge~$e\in E(G)$ when~$c = 3$. The only way these vertices can form an FSG is by including the endpoints of~$e$ as connecting vertices which are both in~$V(G)$. Summarizing, every FSG~$H$ in~$G'$ contains the two endpoints of some edge $e\in E(G)$ and no other vertices from~$V(G)$, as claimed. By construction, for each edge~$e\in E(G)$ we also have \emph{at least} one FSG~$H$ with~$H\cap V(G) = e$. With the help of these two facts, we can now show that
\begin{align*}
	(G,k) \text{ is a yes-instance} \Leftrightarrow (G',k) \text{ is a yes-instance.}
\end{align*}

	($\Rightarrow$) Since each FSG in~$G'$ contains the two endpoints of some edge in~$G$ we know that any solution~$S$ for~$(G,k)$ is also a solution for~$(G',k)$.
	
	($\Leftarrow$) Let~$S'$ be a solution for~$(G',k)$. We now construct a set~$S$ that is a solution for~$(G,k)$. The FSGs in~$G'$ share a vertex if and only if the corresponding edges in~$G$ share an endpoint. This can easily be seen from the construction of~$G'$ by observing that each FSG must have diameter~$2$. That means if a vertex in~$S'$ only appears in one FSG, then we add an arbitrary endpoint of the corresponding edge to~$S$. If a vertex in~$S'$ appears in multiple FSGs, we just add the shared endpoint of the corresponding edges to~$S$. The solution~$S$ now clearly contains at most~$k$ vertices from~$V(G)$ and covers all edges in~$E(G)$.
\end{proof}

\subsection{Polynomial-Time Solvable Cases}

If~$\Delta < c$, then~$\probshort$ is obviously polynomial-time solvable since any input graph is already~$c$-closed. To show less trivial cases, we use the following reduction rule that removes edges that are not critical edges and do not create new FSGs when removed. Recall that an edge is critical in~$G$ if it is between a bad pair vertex and one of its connecting vertices of some FSG.

\begin{rrule}\label{rrule:non_critical_edges}
	Remove an edge~$\{ u,v \}\in E(G)$ if it is not critical in~$G$ and~$|N(u) \cap N(v)| < c$.
\end{rrule}

\begin{lemma}
	Rule~\ref{rrule:non_critical_edges} is correct.
\end{lemma}
\begin{proof}
  Let~$(G,k)$ be a~$\probshort$ instance and let~$G'$ be the graph that is created after Rule~\ref{rrule:non_critical_edges} removed edge~$e$ from~$G$. We show that
  \begin{align*}
	(G,k)\text{ is a yes-instance} \Leftrightarrow (G',k)\text{ is a yes-instance.}
\end{align*}
	
	($\Rightarrow$) Let~$S$ be a solution for~$(G,k)$. Since~$S$ covers all FSGs in~$G$, the only way that~$S$ is not a solution for~$(G',k)$ is if removing~$e$ creates a new FSG~$H$ in~$G'$. Removing an edge cannot add a new common neighbor to a pair of vertices. That means~$e$ is an edge between the bad pair of~$H$. But then the bad pair in~$H$ has at least~$c$ common neighbors in~$G$ which contradicts the precondition of Rule~\ref{rrule:non_critical_edges}. Hence,~$S$ is also a solution for~$(G',k)$.
	
	($\Leftarrow$) Let~$S'$ be a solution for~$(G',k)$. Since~$S'$ covers all FSGs in~$G'$, the only way that~$S'$ is not also a solution for~$(G,k)$ is if removing~$e$ destroys an FSG~$H$ in~$G$. This is only possible if removing~$e$ reduces the number of common neighbors of the bad pair in~$H$. However, that means that~$e$ is an edge between a bad pair vertex and a connecting vertex in~$H$ which contradicts that~$e$ is not critical due to the precondition of Rule~\ref{rrule:non_critical_edges}. Hence,~$S'$ is also a solution for~$(G,k)$.
\end{proof}

Note that this rule can be exhaustively applied in~$\Oh(n)$ time if the maximum degree is constant. This is because such a graph can have at most~$\frac{\Delta}{2} n$ edges and checking the two conditions from Rule~\ref{rrule:non_critical_edges} for an edge can be done in constant time. We only need to check each edge once because the second precondition of Rule~\ref{rrule:non_critical_edges} ensures that no new FSGs are created when removing an edge.

We can now use this rule to show that~$\probshort$ is polynomial-time solvable for~$c = \Delta = 2$ and~$c = \Delta = 3$. In both proofs we use the following observation: if Rule~\ref{rrule:non_critical_edges} cannot be applied to a graph~$G$ with~$c = \Delta$, then every edge in~$G$ must be a critical edge. This is because there cannot be an edge where the two endpoints have at least~$c$ common neighbors as they would then have degree at least~$c + 1$.

\begin{figure}[t]
        \centering
        \begin{subfigure}[t]{0.3\textwidth}
            \centering
            \begin{tikzpicture}[xscale=0.5,yscale=0.3, every node/.style={circle, fill, inner sep=2pt, minimum size=0pt, draw=none}]
            \node[label=left:$w_1$,fill=lightgray] (W) at (0,2) [circle,draw] {};
            \node[label=right:$w_2$,fill=lightgray] (X) at (2,2) [circle,draw] {};
            \node[label=right:$w_3$,fill=lightgray] (Y) at (4,2) [circle,draw] {};
            \node[label=$u$] (U) at (1,4) [circle,draw] {};
            \node[label=$v$] (V) at (3,4) [circle,draw] {};
            \node[label=left:$t$,fill=white] (A) at (1,0) [circle,draw] {};
            \node[label=right:$s$,fill=white] (B) at (3,0) [circle,draw] {};
            
            \draw (U) -- (W);
            \draw (U) -- (X);
            \draw (U) -- (Y);
            \draw (V) -- (W);
            \draw (V) -- (X);
            \draw (V) -- (Y);
            \draw[dashed] (A) -- (W);  
            \end{tikzpicture}
            \caption{\textsc{FSG} with additional dashed critical edge.}
            \label{fig:maxdeg3a}
        \end{subfigure}
        \hfill
        \begin{subfigure}[t]{0.3\textwidth}
            \centering
            \begin{tikzpicture}[xscale=0.5,yscale=0.3, every node/.style={circle, fill, inner sep=2pt, minimum size=0pt, draw=none}]
            \node[label=left:$w_1$,fill=lightgray] (W) at (0,2) [circle,draw] {};
            \node[label=right:$w_2$,fill=lightgray] (X) at (2,2) [circle,draw] {};
            \node[label=right:$w_3$,fill=lightgray] (Y) at (4,2) [circle,draw] {};
            \node[label=$u$] (U) at (1,4) [circle,draw] {};
            \node[label=$v$] (V) at (3,4) [circle,draw] {};
            \node[label=left:$t$,fill=white] (A) at (1,0) [circle,draw] {};
            \node[label=right:$s$,fill=white] (B) at (3,0) [circle,draw] {};
            
            \draw (U) -- (W);
            \draw (U) -- (X);
            \draw (U) -- (Y);
            \draw (V) -- (W);
            \draw (V) -- (X);
            \draw (V) -- (Y);
            \draw (A) -- (W);
            \draw (X) -- (A);
            \draw (A) -- (Y);
            \draw (A) -- (X);
        \end{tikzpicture}
        \caption{Case 2:~$t$ is in a bad pair. All vertices must have maximum degree.}
        \label{fig:maxdeg3b}
        \end{subfigure}
        \hfill
        \begin{subfigure}[t]{0.3\textwidth}
            \centering
            \begin{tikzpicture}[xscale=0.5,yscale=0.3, every node/.style={circle, fill, inner sep=2pt, minimum size=0pt, draw=none}]
            \node[label=left:$w_1$,fill=lightgray] (W) at (0,2) [circle,draw] {};
            \node[label=right:$w_2$,fill=lightgray] (X) at (2,2) [circle,draw] {};
            \node[label=right:$w_3$,fill=lightgray] (Y) at (4,2) [circle,draw] {};
            \node[label=$u$] (U) at (1,4) [circle,draw] {};
            \node[label=$v$] (V) at (3,4) [circle,draw] {};
            \node[label=left:$t$,fill=white] (A) at (1,0) [circle,draw] {};
            \node[label=right:$s$,fill=white] (B) at (3,0) [circle,draw] {};
            
            \draw (U) -- (W);
            \draw (U) -- (X);
            \draw (U) -- (Y);
            \draw (V) -- (W);
            \draw (V) -- (X);
            \draw (V) -- (Y);
            \draw (A) -- (W);
            \draw (X) -- (A);
            \draw[dashed] (A) -- (B);
            \draw[dashed] (B) -- (Y);
        \end{tikzpicture}
        \caption{Case 3:~$w_1$ is in a bad pair. We prove that the dashed edges cannot exist.}
        \label{fig:maxdeg3c}
        \end{subfigure}
        \caption{Connected components for maximum degree-$3$ graphs with~$c = 3$ after exhaustively applying Rule~\ref{rrule:non_critical_edges}. The black vertices are the bad pair in the initial FSG. The gray vertices are the connecting vertices of the bad pair and the white vertices are additional vertices that may or may not be part of the component.}
        \label{fig:maxdeg3}
\end{figure}
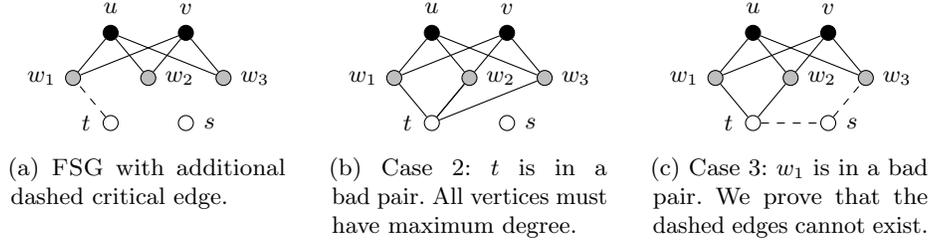

\begin{theorem}\label{lem:maxdeg3_easy}
	$\probshort$ can be solved in~$\Oh(n)$ time on graphs with~$c = \Delta = 2$ and on graphs with~$c = \Delta = 3$.
\end{theorem}
\begin{proof}[Proof for~$c = \Delta = 2$]\label{proof:maxdeg3_easy}
	Let~$(G,k)$ be a~$\probshort$ instance with~$c = \Delta = 2$ and let~$G'$ be the graph that is obtained by exhaustively applying Rule~\ref{rrule:non_critical_edges} to~$G$. We know that every edge in~$G'$ must be part of an FSG. However, an FSG with the smallest possible number of edges for~$c = 2$ has four vertices that all already have degree~$2$. This means any connected component in~$G'$ is either an isolated vertex or such an FSG. Thus, an optimal solution for~$G'$ can be found by removing any single vertex from each FSG. Applying Rule~\ref{rrule:non_critical_edges} can be done in~$\Oh(n)$ time and removing a vertex from each non-singleton connected component can also be done in~$\Oh(n)$ time.
\end{proof}
\begin{proof}[Proof for~$c = \Delta = 3$]
	Let~$(G,k)$ be a~$\probshort$ instance with~$c = \Delta = 3$ and let~$G'$ be the graph obtained by exhaustively applying Rule~\ref{rrule:non_critical_edges} to~$G$. We show that each connected component in~$G'$ has at most~six vertices and can therefore be solved in constant time. The main steps of the proof are illustrated in Figure~\ref{fig:maxdeg3}. Clearly, each connected component with at least one edge must contain an FSG. Let~$H$ be such an FSG with the bad pair~$u$ and~$v$ and the connecting vertices~$w_1, w_2,$ and~$w_3$. The bad pair vertices already have degree~$3$ and can therefore not be adjacent to any other vertex. The connecting vertices have degree~$2$ and are therefore adjacent to at most one further vertex. If no additional edge exists, then the component only has~five vertices and we are done. Otherwise, assume without loss of generality that~$w_1$ is adjacent to a vertex~$t$ (see Figure~\ref{fig:maxdeg3a}). The edge~$\{w_1, t\}$ must be critical which leads to three possible cases.
	
	\textbf{Case $1$}: $t$ is equal to~$w_2$ or~$w_3$. Without loss of generality assume~$t = w_2$. To further increase the size of the component~$w_3$ must now connect to a new vertex via a bridge edge. However, any edge in an FSG is part of at least one cycle which means a bridge edge cannot be critical.
	
	\textbf{Case $2$}: $t$ is part of a bad pair. Since~$c = \Delta$ we know that all neighbors of a bad pair vertex must be connecting vertices. Here this means that~$w_1$ is a connecting vertex. Consequently, the other bad pair vertex is either~$u$ or~$v$ since it must be adjacent to~$w_1$. But then~$w_2$ and~$w_3$ are also connecting vertices and~$t$ is adjacent to them. This leads to a graph with six~vertices all of which have degree~$3$ (see Figure~\ref{fig:maxdeg3b}).
	
	\textbf{Case $3$}: $w_1$ is part of a bad pair and~$t$ is not. Since the other vertex of the bad pair must be connected to all three neighbors of~$w_1$ that vertex can only be~$w_2$ or~$w_3$. Without loss of generality assume that it is~$w_2$. This leaves~$w_3$ and~$t$ as the only vertices with degree less than~$3$ (see Figure~\ref{fig:maxdeg3c}). If we connect them with an edge we again have a graph with~$6$ vertices all of which have degree~$3$. Thus, to increase the size of the component to more than~$6$ at least one of~$w_3$ and~$t$ must be connected to a seventh vertex~$s$ by an edge. This new edge again implies that one of the two endpoints must be in a bad pair. We already showed that~$t$ being part of a bad pair leads to a component of size~$6$. And if~$w_3$ was part of a bad pair, then we would need to connect~$s$ to~$w_1$ or~$w_2$ as these are the only vertices that could be the second vertex of the bad pair. But that is not possible because these vertices already have degree~$3$.
This only leaves the case where~$s$ is part of a bad pair. If~$s$ is adjacent to~$w_3$, then the other vertex of the bad pair must be~$u$ or~$v$ which is not possible since~$s$ cannot be adjacent to~$w_1$ or~$w_2$. If~$s$ is adjacent to~$t$, then~$w_1$ or~$w_2$ must be the other vertex of the bad pair. But that is also not possible because~$s$ cannot be adjacent to~$u$ or~$v$. Hence, we cannot connect to a seventh vertex and each connected component contains at most six~vertices.
	
	The running time can be seen as follows: We can apply Rule~\ref{rrule:non_critical_edges} in~$\Oh(n)$ time on graphs with constant maximum degree. Then we can go through the~$\Oh(n)$ connected components and solve each in constant time.
\end{proof}

\section{On Problem Kernelization}
\label{sec:kernel}

In this section, we first show that~$\probshort$ does not admit a kernel of bit size~$\Oh(k^{c-\epsilon})$ unless \textsc{coNP}~$\subseteq$ \textsc{NP}/poly using a reduction from \textsc{$c$-Hitting Set}. We then derive a problem kernel for~$\probshort$ with~$\Oh(k^{c+2})$ vertices using a reduction to \textsc{$c + 2$-Hitting Set}.
Finally, we show a problem kernel of size~$\Oh(x^2 \cdot (k+c))$ where~$x$ is the number of vertices that are in at least one bad pair.

\subsection{Hitting Set Reductions}

We start by defining \textsc{$d$-Hitting Set}:
\probdefdec
{\textsc{$d$-Hitting Set}}
{A family~$\mathcal{A}$ of subsets of a universe~$U$ such that each subset has size at most~$d$, and a positive integer~$k$.}
{Is there a set~$H \subseteq U$ with~$|H| \leq k$ such that~$H$ contains at least one element of each set in~$\mathcal{A}$?}

We now define a reduction from \textsc{$d$-Hitting Set} to~$\probshort$ with~$c = d$ by turning the elements in each set into the connecting vertices of an FSG.

Let~$(U,\mathcal{A},k)$ be an instance of \textsc{$d$-Hitting Set}. We first add additional elements to the sets in~$\mathcal{A}$ such that each set has exactly size~$d$. We construct the graph~$G$ as follows: Each element of~$U$ becomes a vertex in~$G$. Additionally, for each set~$A\in \mathcal{A}$ we add two new vertices~$v_A$ and~$u_A$. For each vertex~$v\in A$ we add the edges~$\{v, v_A \}$ and~$\{v, u_A \}$ such that~$v_A$ and~$u_A$ form a bad pair with the vertices in~$A$ as their connecting vertices. Finally, we add all edges between vertices in~$U$ such that~$U$ becomes a clique. With this we obtain the~$\probshort$ instance~$(G,k)$ with~$c = d$.

\begin{lemma}\label{lem:from_hit_set_red_correct}
The reduction from \textsc{$d$-Hitting Set} to~$\probshort$ is correct.
\end{lemma}
\begin{proof}\label{proof:from_hit_set_red}
We need to show the following:
\begin{align*}
	(U,\mathcal{A},k) \text{ is a yes-instance} \Leftrightarrow (G,k) \text{ is a yes-instance.}
\end{align*}

$(\Rightarrow)$ Let~$S\subseteq U$ be a solution for~$(U,\mathcal{A},k)$. Any FSG~$H$ in~$G$ must have at least one bad pair vertex~$v_A$ that is not in~$U$ since all vertices in~$U$ are pairwise connected. By construction, we have~$N(v_A) = A \subseteq V(H)$. Since~$S$ contains at least one vertex of~$A$, we know that~$G - S$ does not contain any FSGs. This means that~$S$ is a solution for~$(G,k)$.

$(\Leftarrow)$ Let~$S\subseteq V(G)$ be a solution for~$(G,k)$. We construct the set~$S'$ as follows: First we add any vertex in~$S$ that is also in~$U$ to~$S'$. Then for each vertex~$v$ in~$S$ that is not in~$U$ there are two cases. In the first case,~$v$ is a vertex that was added to a set~$A\in \mathcal{A}$ to increase its size to~$d$. In this case we add an arbitrary vertex from~$A$ to~$S'$. In the second case,~$v$ is a vertex that was added to create a bad pair with~$A\in \mathcal{A}$ as the connecting vertices. In this case we also add an arbitrary vertex in~$A$ to~$S'$. Clearly we have~$|S'| \leq k$.
For each set~$A\in \mathcal{A}$ there is at least one FSG in~$G$ and therefore there must be at least one vertex~$v\in S$ that is contained in that FSG. By construction, at least one vertex of~$S'$ is contained in~$A$. Hence,~$S'$ is a solution for~$(U,\mathcal{A},k)$.
\end{proof}

From this reduction we can obtain two results. First, the constructed graph is a split graph which means that~$\probshort$ is NP-hard on split graphs.

\begin{corollary}
	For every~$c \geq 2$,~$\probshort$ is NP-hard on split graphs.
\end{corollary}
Let us remark that the case~$c=1$ can be solved in polynomial time on split graphs~\cite{CKOY18}.
The second result is a lower bound for the bit size of a problem kernel for~$\probshort$. For this we can use the result that \textsc{$d$-Hitting Set} does not have a compression with bit size~$\Oh(k^{d - \epsilon})$ for any~$\epsilon > 0$ unless \textsc{coNP}~$\subseteq$ \textsc{NP}/poly~\cite{PA16}. Since~$k$ does not change in the reduction and we have~$c = d$ we immediately obtain the following result for~$\probshort$.

\begin{theorem}\label{lem:kernel_lb}
	$\probname$ with~$c \geq 2$ does not have a kernel with bit size~$\Oh(k^{c - \epsilon})$ for any~$\epsilon > 0$ unless \textsc{coNP}~$\subseteq$ \textsc{NP}/poly.
\end{theorem}

Next, we want to show that~$\probshort$ admits a kernel with~$\Oh(k^{c + 2})$ vertices. For this we first prove the following lemma.

\begin{lemma}\label{lem:many_conn_verts}
	Let~$(G,k)$ be an instance of~$\probshort$ and let~$u, v$ be a bad pair in~$G$ with at least~$k + c$ connecting vertices. Any solution for~$(G,k)$ contains at least one of~$u$ and~$v$.
\end{lemma}
\begin{proof}
  Assume that~$(G,k)$ has a solution~$S$ that does not contain~$u$ or~$v$. Then,~$u$~and~$v$ still have at least~$c$ connecting vertices in~$G - S$ as~$S$ contains at most~$k$ connecting vertices of~$u$ and~$v$. Hence,~$u$ and~$v$ form an FSG with these connecting vertices in~$G - S$. This contradicts~$S$ being a solution for~$(G,k)$.
\end{proof}

We now define a reduction from~$\probshort$ to \textsc{$d$-Hitting Set} with~$d = c + 2$ by turning (most of) the FSGs into sets of the \textsc{$d$-Hitting Set} instance.
Let~$(G,k)$ be an instance of~$\probshort$. The universe~$U$ consists of all vertices~$V(G)$. For each bad pair~$u, v$ in~$G$ we do the following: If~$u$ and~$v$ have less than~$k + c$ common neighbors, then, for each FSG where~$u, v$ is the bad pair, we add a new set to~$\mathcal{A}$ that contains all vertices in the FSG. Otherwise, we choose an arbitrary subset~$C$ containing exactly~$k + c$ of these common neighbors. Then, for each FSG where~$u, v$ is the bad pair and the connecting vertices are all contained in~$C$, we add a new subset to~$\mathcal{A}$ that contains all vertices in the FSG.
This gives the \textsc{$d$-Hitting Set} instance~$(U,\mathcal{A},k)$ with~$d = c + 2$.

\begin{lemma}\label{lem:to_hit_set_red}
The reduction from~$\probshort$ to \textsc{$d$-Hitting Set} is correct.
\end{lemma}
\begin{proof}\label{proof:to_hit_set_red}
We need to show the following:
\begin{align*}
	(G,k) \text{ is a yes-instance} \Leftrightarrow (U,\mathcal{A},k) \text{ is a yes-instance.}
\end{align*}

($\Rightarrow$) Let~$S\subseteq V(G)$ be a solution for~$(G,k)$.Then,~$S$ contains at least one vertex from each FSG in~$G$ and since each set in~$\mathcal{A}$ corresponds to the vertex set of some FSG in~$G$ we know that~$S$ must also be a solution for~$(U,\mathcal{A},k)$.

($\Leftarrow$) Let~$S\subseteq U$ be a solution for~$(U,\mathcal{A},k)$. We show that~$S$ is also a solution for~$(G,k)$. Let~$H$ be an FSG in~$G$ with the bad pair~$u, v$. If~$u, v$ have less than~$k + c$ common neighbors, then~$\mathcal{A}$ contains an edge with the vertex set of~$H$ which means that~$H$ is covered by~$S$. If~$u, v$ have at least~$k + c$ common neighbors, then we can use a similar argument as in Lemma~\ref{lem:many_conn_verts}. We cannot cover all sets in~$\mathcal{A}$ that correspond to FSG where~$u, v$ is the bad pair with only~$k$ vertices, unless at least one of those vertices is~$u$ or~$v$. That means~$S$ must contain~$u$ or~$v$ and it therefore covers~$H$.
\end{proof}

Van Bevern showed that \textsc{$d$-Hitting Set} has an \emph{expressive} kernel with~$\Oh(k^d)$ elements and sets~\cite{DBLP:journals/algorithmica/Bevern14}. Expressive kernels are defined as follows:
\begin{definition}[\cite{DBLP:journals/algorithmica/Bevern14}]\label{def:expressive_kernel}
	A kernelization algorithm for \textsc{$d$-Hitting Set} is \emph{expressive} if, given an instance~$(U,\mathcal{A},k)$, it outputs an instance~$(U',\mathcal{A}',k')$ such that
	\begin{enumerate}
		\item $U' \subseteq U$ and~$\mathcal{A}' \subseteq \mathcal{A}$,
		\item any vertex set of size at most~$k$ is a minimal hitting set for~$(U,\mathcal{A})$ if and only if it is a minimal hitting set for~$(U',\mathcal{A}')$, and 
		\item  it outputs a certificate for~$(U',\mathcal{A}',k')$ being yes if and only if~$(U,\mathcal{A},k)$ is.
	\end{enumerate}
\end{definition}
We can now use this result to obtain the following kernel for~$\probshort$.

\begin{theorem}
	For~$c \geq 2$,~$\probname$ admits a kernel with $\Oh(k^{c + 2})$ vertices that can be computed in~$\Oh(n^3 + n^2 m)$ time.
\end{theorem}
\begin{proof}
	Let~$(G,k)$ be a~$\probshort$ instance, let~$(U,\mathcal{A},k)$ be the \textsc{$c + 2$-Hitting Set} instance constructed with the reduction described above, and let~$(U',\mathcal{A}',k')$ be the instance obtained by using the kernelization algorithm of Van Bevern~\cite{DBLP:journals/algorithmica/Bevern14}. Since~$k$ does not change in the kernelization, we can replace~$k'$ by~$k$. Our kernel is now the~$\probshort$ instance~$(G',k)$ that is constructed by removing any vertex in~$V(G)\setminus U'$ and their incident edges from~$G$.
	
We first prove equivalence of the instances, that is, we show that
\begin{align*}
	(G,k) \text{ is a yes-instance} \Leftrightarrow (G',k) \text{ is a yes-instance.}
\end{align*}

($\Rightarrow$) Let~$S$ be a solution for~$(G,k)$. Since the property of being~$c$-closed is hereditary and we only deleted vertices to obtain~$G'$ from~$G$, it is clear that~$S$ is also a solution for~$(G',k)$.

($\Leftarrow$) Let~$S$ be a minimal solution for~$(G',k)$. Each set in~$H\in \mathcal{A}'$ also exists in~$\mathcal{A}$ by the first property in Definition~\ref{def:expressive_kernel}. Hence,~$H$ is the vertex set of some FSG in~$G$. Since we did not remove any vertices in~$H$ when creating~$G'$, we know that~$H$ is also the vertex set of some FSG in~$G'$. This means that~$S$ is a solution for~$(U',\mathcal{A}',k)$. According to the second property in Definition~\ref{def:expressive_kernel},~$S$ is also a solution for~$(U,\mathcal{A},k)$. Finally, the proof of Lemma~\ref{lem:to_hit_set_red} shows that~$S$ is also a solution for~$(G,k)$.

We have now shown that~$(G',k)$ is equivalent and since the vertex set~$V(G')$ is equal to~$U'$ we know that~$G'$ has~$\Oh(k^{c + 2})$ vertices. It remains to bound the running time of the kernelization. The kernelization consists of three main steps. First, the reduction from~$(G,k)$ to~$(U,\mathcal{A},k)$. Second, the expressive kernelization of~$(U,\mathcal{A},k)$. Finally, creating the induced subgraph~$G' = G[U']$. Observe that if~$n \leq k^{c + 2}$, then the kernelization can simply output~$(G,k)$. Hence, we assume~$n > k^{c + 2}$ in the following.

For the first step, we start by looking at each pair of vertices~$u$ and~$v$ and enumerating their common neighbors in~$\Oh(n\cdot m)$ time. If~$u$ and~$v$ have less than~$c$ common neighbors we do not need to do anything. Otherwise, we construct set~$C$ of at most~$k+c$ common neighbors, and then add~$\Oh((k + c)^c)$ sets of size~$c + 2$ to~$\mathcal{A}$. Since $c$ is constant and we can assume that~$n > k^{c + 2}$ we can say that we only add~$\Oh(n)$ sets of constant size to~$\mathcal{A}$ for each pair of vertices in~$G$. This bounds the size of~$\mathcal{A}$ to~$\Oh(n^3)$ and the running time of the first step to~$\Oh(n^2 m)$.
The second step can be done in~$\Oh(n + |\mathcal{A}|)$ time for constant~$d$~\cite{DBLP:journals/algorithmica/Bevern14}. Combined with the bound on~$|\mathcal{A}|$ from above we get a running time of~$\Oh(n^3)$ for the second step.
The final step can be done in~$\Oh(n + m)$ time but since both terms are dominated by the running times of the first two steps we get the total running time of~$\Oh(n^3 + n^2 m)$.
\end{proof}

\subsection{Kernel for the Number of Bad Pair Vertices}

To obtain a kernel with polynomial dependence on~$c$, we use a new parameter~$x:=|X|$ where~$X$ is the set of vertices in~$V(G)$ that are part of at least one bad pair in~$G$. 
With this parameter we can obtain a kernel with~$\Oh(x^2 \cdot (k+c))$ vertices. To obtain this kernel we first introduce a new reduction rule and then show that exhaustively applying this rule to a graph~$G$ will reduce its size to the desired bound.
We may assume~$x > k$ since any instance with~$x \leq k$ can be trivially solved by removing all vertices of~$X$.
 We say that a bad pair is \emph{heavy} if it has more than~$k + c$ connecting vertices in~$V(G)\setminus X$. Otherwise, the bad pair is \emph{weak}.

\begin{rrule}\label{rrule:unmarked}
	Mark all vertices  in~$X$ and  all vertices that are connecting vertices of at least one weak bad pair. Then, remove an arbitrary unmarked vertex.
\end{rrule}

\begin{lemma}
	Rule~\ref{rrule:unmarked} is correct.
\end{lemma}
\begin{proof}
Let~$(G,k)$ be an instance of~$\probshort$ and let~$(G',k)$ be the instance that is obtained by Rule~\ref{rrule:unmarked} removing vertex~$v$ from~$G$. We need to show that
\begin{align*}
	(G,k) \text{ is a yes-instance} \Leftrightarrow (G',k) \text{ is a yes-instance.}
\end{align*}

$(\Rightarrow)$ Let~$S$ be a solution for~$(G,k)$ which means that~$G - S$ is~$c$-closed. Since the property of being~$c$-closed is hereditary, we know that~$(G - S) - \{ v \} = (G' - S)$ is also~$c$-closed. Hence,~$S\setminus \{ v \}$ is a solution for~$(G',k)$.

$(\Leftarrow)$ Let~$S'$ be a solution for~$(G',k)$. Assume towards a contradiction that~$G - S'$ is not~$c$-closed and therefore contains some FSG~$H$. If~$v\notin V(H)$, then~$H$ must also be a subgraph of~$G' - S'$ which contradicts~$S'$ being a solution for~$(G',k)$. If~$v\in V(H)$, then~$v$ must be a connecting vertex in~$H$ since~$v$ is not in~$X$. Moreover, the bad pair in~$H$ is a heavy bad pair in~$G$ since~$v$ was not marked by Rule~\ref{rrule:unmarked}. This means that the bad pair in~$H$ must be a bad pair in~$G'$ with at least~$k + c$ connecting vertices. Lemma~\ref{lem:many_conn_verts} now tells us that~$S'$ contains a vertex from this bad pair. This contradicts~$H$ being an FSG in~$G - S'$.
\end{proof}

Exhaustively applying Rule~\ref{rrule:unmarked} leads to a kernel of the desired size.

\begin{lemma}\label{lem:kernel_x}
	$\probname$ has a problem kernel with~$\Oh(x^2 \cdot (k + c))$ vertices where~$x$ is the number of vertices that are part of at least one bad~pair.
\end{lemma}
\begin{proof}
	Let~$(G,k)$ be an instance of~$\probshort$ such that Rule~\ref{rrule:unmarked} cannot be applied anymore. This means that all vertices in~$G$ are marked and are therefore part of a bad pair or connecting vertices of at least one weak bad pair.
	Graph~$G$ can contain at most~$\binom{x}{2} = \frac{x\cdot (x-1)}{2} \in \Oh(x^2)$ weak bad pairs. Each of those can have at most~$k + c$ connecting vertices. Since~$|X| = x$ and each vertex not in~$X$ is a connecting vertex of at least one weak bad pair, we obtain~$\Oh(x^2 \cdot (k + c))$ as an upper bound for the number of vertices in~$G$.
\end{proof}

Since we can assume~$x > k$, we can also write the bound of Lemma~\ref{lem:kernel_x} as~$\Oh(x^3+x^2 \cdot c)$.

\section{Easy Special Cases}
\label{sec:easy}
\subsection{Interval Graphs}
A graph~$G$ is an \emph{interval graph} if each vertex~$v$ corresponds to a closed interval~$I_v=[\ell(v),r(v)]$ on the real line and two vertices are connected with an edge if and only if their intervals overlap. A \emph{unit interval graph} is an interval graph where all intervals have the same length.  The \emph{depth} of an interval graph is the largest number of intervals that overlap in the same point. The depth is the same as the size of a largest clique in the graph. We assume that~$(v_1, v_2, \hdots, v_n)$ is an ordering of~$V(G)$ with nondecreasing starting points.

We now describe an algorithm that solves~$\probshort$ on unit interval graphs with depth at most~$c + 1$ in linear time~$\Oh(c\cdot n)$. For this we first analyze how FSGs look in such graphs. Figure~\ref{fig:intervalFSGex} shows an example FSG for~$c = 3$.

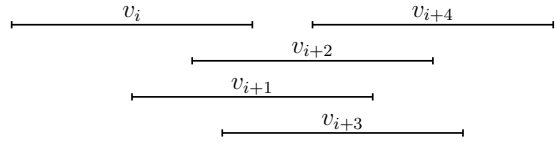
\begin{figure}[t]
    \centering
    \scalebox{1.0}{
    \begin{tikzpicture}[yscale=0.6]
    	\draw[thick] (0, 3) -- (4, 3);
    	\draw[thick] (0, 3.1) -- (0,2.9);
    	\draw[thick] (4, 3.1) -- (4,2.9);
        \node at (2,3.3) {\large $v_i$};
    	\draw[thick] (5, 3) -- (9, 3);
    	\draw[thick] (5, 3.1) -- (5,2.9);
    	\draw[thick] (9, 3.1) -- (9,2.9);
        \node at (7,3.3) {\large $v_{i + 4}$};
    	\draw[thick] (3, 2) -- (7, 2);
    	\draw[thick] (3, 2.1) -- (3,1.9);
    	\draw[thick] (7, 2.1) -- (7,1.9);
        \node at (5,2.3) {\large $v_{i + 2}$};
    	\draw[thick] (2, 1) -- (6, 1);
    	\draw[thick] (2, 1.1) -- (2,0.9);
    	\draw[thick] (6, 1.1) -- (6,0.9);
        \node at (4,1.3) {\large $v_{i + 1}$};
    	\draw[thick] (3.5, 0) -- (7.5, 0);
    	\draw[thick] (3.5, 0.1) -- (3.5,-0.1);
    	\draw[thick] (7.5, 0.1) -- (7.5,-0.1);
        \node at (5.5,0.3) {\large $v_{i + 3}$};
    \end{tikzpicture}
    }
    \caption{Example FSG for~$c = 3$ in a unit interval graph with depth~$c + 1 = 4$. The vertices~$v_i$ and~$v_{i + 4}$ form the bad pair while the other three vertices are the connecting vertices.}\label{fig:intervalFSGex}
\end{figure}

\begin{lemma}\label{lem:FSGstructure}
	Let~$H$ be an FSG in a unit interval graph~$G$ with depth~$\leq c + 1$. We then have~$V(H) = \{ v_i, v_{i + 1}, \hdots, v_{i + c}, v_{i + c + 1} \}$ for some~$1\leq i\leq n$ where~$v_i$ and~$v_{i + c + 1}$ form the bad pair and~$v_{i + 1}, \hdots, v_{i + c}$ are the connecting vertices.
\end{lemma}
\begin{proof}\label{proof:FSGstructure}
	Let~$v_i$ be the vertex of the bad pair in~$H$ that has the smaller starting point. We know that they cannot have the same starting point because their intervals do not overlap. Let~$v_j$ be the other vertex of the bad pair. We know that the intervals of the connecting vertices must overlap with both~$v_i$ and~$v_j$. Since all intervals have the same length and~$r(v_i) < \ell(v_j)$ we know that each connecting vertex~$v_h$ must start after~$\ell(v_i)$ and end before~$r(v_j)$.
	
	Since we must also have~$\ell(v_h) \leq r(v_i)$ we know that the starting points of each connecting vertex are between the starting points of the bad pair vertices, that is,~$\ell(v_i) < \ell(v_h) < \ell(v_j)$ for each connecting vertex~$v_h$. To prove the lemma we now just need to show that besides exactly~$c$ connecting vertices there are no other vertices that have a starting point between the starting points of the bad pair.
	
	We know that there are at least~$c$ connecting vertices. We also know that the interval of each connecting vertex must contain the point~$r(v_i)$ and the point~$\ell(v_j)$. That means for each of those points we have exactly~$c + 1$ intervals containing that point. This shows that there cannot be more than~$c$ connecting vertices. Now let us assume towards a contradiction that there is another interval~$v_h$ with~$\ell(v_i) < \ell(v_h) < \ell(v_j)$ that is not a connecting vertex of the bad pair~$v_i$,~$v_j$. We know that~$v_h$ cannot simultaneously contain the points~$r(v_i)$ and~$\ell(v_j)$. However, the only way this would be possible is if we have~$\ell(v_i) < \ell(v_h)$ and~$r(v_h) < \ell(v_j)$. That is a contradiction because the distance between~$r(v_i)$ and~$\ell(v_j)$ must be at most the length of an interval because the connecting vertices overlap with both.
\end{proof}

We can now solve~$\probshort$ on unit interval graphs with depth at most~$c + 1$ with a simple greedy algorithm: For each vertex~$v_i$, check if~$v_i$ is the first vertex of the bad pair in an FSG. If that is the case, then add the last vertex of the FSG to the solution set.

\begin{algorithm2e}[t]
    \caption{Interval graphs}\label{algorithm:1}
    \KwIn{A unit interval graph~$G$ with depth~$\leq c + 1$.}
    \KwOut{A set~$S$ such that~$G-S$ is~$c$-closed.}

    $S \gets \emptyset$\;
    $i \gets 1$\;

    \While{$i < n$}{
        \If{checkFSG($G, v_i$) = true}{
            $S \gets S \cup \{v_{i+c+1}\}$\;
            $i \gets i+c+2$\;
        }
        \Else{
            $i \gets i + 1$\;
        }
    }
    \Return $S$\;

    \BlankLine
    \SetKwFunction{FCheckFSG}{checkFSG}
    \SetKwProg{Fn}{Function}{:}{}
    \Fn{\FCheckFSG{$G, v_i$}}{
    	\If{$i + c + 1 > n$}{
			\Return false\;    	
    	}
        $h \gets 1$\;
        \While{$h \leq c$}{
            \If{$\ell(v_{i+h}) \leq r(v_i)$ \textbf{and} $r(v_{i+h}) \geq \ell(v_{i+c+1})$}{
                $h \gets h + 1$\;
            }
            \Else{
                \Return false\;
            }
        }
        \Return true\;
    }
\end{algorithm2e}

\begin{theorem}\label{thm:unit-inter}
  \probname~can be solved in~$\Oh(n\cdot c)$~time when the input is a unit interval graph with maximal clique size~$c+1$.
\end{theorem}
\begin{proof}\label{proof:unit-inter}
  The pseudocode of the algorithm is shown in Algorithm~\ref{algorithm:1}.
  We show that this algorithm returns the smallest set~$S$ such that~$G - S$ is~$c$-closed.

  First we show that the function \texttt{checkFSG}($G, v_i$) correctly determines if~$v_i$ is the leftmost vertex of an FSG. The function goes through the next~$c$ vertices and checks if they overlap with~$v_i$ and~$v_{i + c + 1}$. If that is the case, then we do not need to check that~$v_i$ and~$v_{i + c + 1}$ do not overlap because the graph has depth at most~$c + 1$ and them overlapping would imply a depth of at least~$c + 2$. This together with Lemma~\ref{lem:FSGstructure} means that the function is correct.
	
	Now let~$S$ be the solution returned by Algorithm~\ref{algorithm:1} for a~$\probshort$ instance~$(G,k)$ where~$G$ is a unit interval graph of depth at most~$c + 1$. We show that~$G - S$ is~$c$-closed. Let us assume towards a contradiction that~$G - S$ contains an FSG~$H$ and let~$v_i$ be the leftmost vertex in~$H$. If the algorithm calls \texttt{checkFSG}($G, v_i$), then it would have removed vertex~$v_{i + c + 1}$ and~$H$ would not be an FSG in~$G - S$. The only time the algorithm does not call \texttt{checkFSG}($G, v_i$) for some vertex is if Line~$6$ was executed. For that to happen the function must have returned true for a vertex~$v_j$ with~$i - c - 1 \leq j < i$. However, that would mean that vertex~$v_{j + c + 1}$ is in~$S$ which leads to a contradiction since~$v_{j + c + 1}\in H$ means that~$H$ would not be an FSG.
	
	Finally, we show that~$S$ is the smallest set such that~$G - S$ is~$c$-closed. The algorithm skips to vertex~$v_{i + c + 2}$ after finding an FSG starting at~$v_i$ and adding~$v_{i + c + 1}$ to~$S$. This means that the~$|S|$ FSGs that lead the algorithm to adding a vertex to~$S$ do not overlap with each other. This gives us a lower bound of~$|S|$ for the number of vertices that need to be removed since we always need to remove at least one vertex from each FSG. Hence,~$S$ is optimal.
\end{proof}

\subsection{Neighborhood Diversity}

 The \emph{neighborhood diversity} of a graph, $\nd(G)$, is the number of \emph{neighborhood classes} of~$G$ which are the equivalence classes of the relation~$\sim$ with~$u\sim v:\Leftrightarrow N(u)\setminus \{v\} = N(v)\setminus \{u\}$. We show that $\probshort$ with $c \geq 2$ is FPT for the parameter $\nd$. To do this we first show that for each neighborhood class there are essentially only~$c + 1$ possibilities to consider.

\begin{lemma}\label{lem:nd_cases}
	Let~$S$ be a minimum-size solution for a~$\probshort$ instance~$(G,k)$ with~$c \geq 2$. Then for each neighborhood class~$D$ we either have~$D\setminus S = D$ or~$|D\setminus S| < c$.
\end{lemma}
\begin{proof}
  Assume towards a contradiction that~$D\setminus S$ is a proper subset of~$D$ that contains at least~$c$ vertices and let~$v$ be a vertex of~$D\cap S$. Since~$S$ is a minimal solution,~$S\setminus \{v\}$ is not a solution, that is,~$G-(S\setminus \{v\})$ contains an FSG~$H$ one of whose vertices is~$v$. Including~$v$, this FSG can contain at most~$c$ vertices of~$D$, since the two bad pair vertices have different neighborhoods than the connecting vertices.   Hence,~$v$ can be replaced in~$H$ by some vertex of~$D\setminus S$, giving an FSG~$H'$ which is contained in~$G-S$. This contradicts that~$S$ is a solution.
\end{proof}

With this we can now obtain the following via branching.

\begin{theorem}\label{lem:nd+c_fpt}
	$\probshort$ with~$c \geq 2$ can be solved in~$\Oh((c + 1)^{\nd} \cdot \nd^2 \cdot m)$ time.
\end{theorem}
\begin{proof}\label{proof:nd+c_fpt}
	Since the vertices in a neighborhood class have the same neighborhood it does not matter which vertices we specifically remove. It only matters how many of them we remove. According to Lemma~\ref{lem:nd_cases} there are only~$c + 1$ different numbers of vertices that need to be removed from each neighborhood class and there are only~$nd$ neighborhood classes. That gives us a total of~$(c + 1)^{nd}$ solutions that we need to check. To check if a possible solution~$S$ solves the problem we only need to check if~$|S| \leq k$ and if~$G - S$ is~$c$-closed. This can be done in time~$\Oh(\nd^2 \cdot m)$ by just checking the common neighborhood of each pair of non empty neighborhood classes in~$G - S$ that are not connected. Performing this check for all~$(c + 1)^{nd}$ solutions gives us an algorithm with the desired running time.
\end{proof}

To avoid $c$ appearing in the running time we formulate the optimization version of~$\probshort$ where we want to find the smallest set~$S$ such that~$G - S$ is~$c$-closed as an ILP with~$\Oh(\nd)$ variables which implies fixed-parameter tractability with respect to~$\nd$ alone. The ILP relies mainly on the observation that we only care about the number of vertices that are removed from each neighborhood class and not about the specific vertices.
\begin{theorem}\label{lem:nd_fpt}
	$\probshort$ can be solved in~$f(\nd) \cdot n^{\Oh(1)}$ time.
\end{theorem}
\begin{proof}\label{proof:nd_fpt}
Let~$\mathcal{D}$ be the set of all neighborhood classes in~$G$. We let~$[v]$ denote the neighborhood class that contains the vertex~$v$. We first introduce a variable~$x_D$ for each~$D\in \mathcal{D}$ that counts the number of vertices that remain in~$D$ after removing the vertices of the solution~$S$. Our goal is to maximize~$\sum_{D\in \mathcal{D}} x_D$.
To make sure that the solution set removes all FSGs from the graph we introduce some additional variables and constraints.
For each neighborhood class~$D\in \mathcal{D}$ we introduce the Boolean variables~$x'_D$ and~$x''_D$. We want~$x'_D$ to be~$0$ if~$x_D$ is~$0$ and~$1$ if~$x_D > 0$. We want~$x''_D$ to be~$0$ if~$x_D \leq 1$ and~$1$ if~$x_D > 1$. For this we use a large constant~$M \gg n^2$ and introduce the constraints~$M\cdot x'_D \geq x_D$,~$x'_D \leq x_D$ and~$M\cdot x''_D + 1 \geq x_D$,~$2 \cdot x''_D \leq x_D$ respectively for each~$D\in \mathcal{D}$.

We now add a constraint for each pair of vertices~$u, v$ that are not connected. For this we define the set~$\mathcal{N}(v)$ as the set of neighborhood classes that contain neighbors of~$v$. If~$[u] \not= [v]$, then we add the constraint
\begin{align*}
	M \cdot (x'_{[u]} + x'_{[v]}) + \left(\smashoperator[r]{\sum_{D\in \mathcal{N}(u) \cap \mathcal{N}(v)}} x_D \,\,\, \right) - c < 2M \text{.}
\end{align*}
The sum counts the total number of common neighbors of~$u$ and~$v$. If either of the two classes~$[u]$ or~$[v]$ is empty then the constraint is automatically fulfilled since~$M$ is much larger than the maximum value of the sum. If~$[u]$ and~$[v]$ are both not empty, then the sum must be smaller than~$c$ for the constraint to be fulfilled.

If~$[u] = [v]$ and~$[u]$ is an independent set we instead add the constraint
\begin{align*}
	M \cdot x''_{[u]} + \left(\smashoperator[r]{\sum_{D\in \mathcal{N}(u)}} x_D\right) - c < M \text{.}
\end{align*}
The sum again counts the common neighbors of~$u$ and~$v$. If~$[u]$ has at most one vertex, then the constraint is fulfilled because~$x''_{[u]}$ would be~$0$. In this case no bad pair can be formed. If~$[u]$ contains at least~$2$ vertices, then a bad pair could be formed and the sum must be smaller than~$c$ for the constraint to be fulfilled.

In total we have the following ILP,~where~$\mathcal{D}_I$ denotes the neighborhood classes that are independent sets.

\begin{align*}
  \max & \sum_{D \in \mathcal{D}} x_D\\
\text{s.t.}~~ &	M \cdot x'_D \geq x_D\geq x'_D && \forall D \in \mathcal{D},\\
    & M \cdot x''_D + 1 \geq x_D\geq 2\cdot x''_D && \forall D \in \mathcal{D}, \\
    & M \cdot (x'_{[u]} + x'_{[v]}) + \left(\smashoperator[r]{\sum_{D\in \mathcal{N}(u) \cap \mathcal{N}(v)}} x_D \,\,\, \right) - c < 2M && \forall [u] \not = [v] : \{u, v\}\notin E(G),  \\
    & M \cdot x''_{[u]} + \left(\smashoperator[r]{\sum_{D\in \mathcal{N}(u)}} x_D\right) - c < M && \forall [u] \in \mathcal{D}_I  
\end{align*}

In this ILP formulation, the number of variables is~$3 \cdot \nd$ which leads to a running time of~$f(\nd) \cdot n^{\Oh(1)}$~\cite{L83}.
\end{proof}

\section{Conclusion}
We have introduced \probname~and provided a first overview of its complexity. Several interesting questions are left open: First, it remains to settle the remaining five open cases for the complexity on bounded-degree graphs. Second, one could close the gap between the upper and lower bounds on the kernel size for the solution size parameter~$k$. Third, a complexity classification for interval graphs and unit interval graphs remains open. 
In addition, it is of interest to identify further problems which can be solved efficiently on almost $c$-closed graphs. Finally, we introduced the bad pairs-related parameter which could be useful for other vertex deletion problems where the forbidden subgraphs are large  but contain distinguished pairs of vertices. Identifying such problems and providing a generic definition of a bad pair-related parameter could be fruitful. 

\bibliographystyle{plain}   
\bibliography{paper}

\end{document}